%% file: root.tex
\documentclass[letterpaper,journal, 10 pt]{IEEEtran} 

\input{usepackages_letter.tex}
\input{macros.tex}
\usepackage{balance}

\usepackage{titlesec} 
\titlespacing*{\section}{0pt}{*0.4}{*0.3}
\titlespacing*{\subsection}{0pt}{*0.35}{*0.25}
\makeatletter
\def\thm@space@setup{%
 \thm@preskip=1pt plus 1pt minus 1pt%
 \thm@postskip=1pt plus 1pt minus 1pt%
}

\renewenvironment{proof}[1][\proofname]{\par
  \pushQED{\qed}%
  \normalfont \topsep=2pt plus 1pt minus 1pt\relax
  \trivlist
  \item[\hskip\labelsep\bfseries #1\@addpunct{.}]\ignorespaces
}{%
  \popQED\endtrivlist\@endpefalse
}
\makeatother

\AtBeginDocument{%
 \setlength{\textfloatsep}{0pt plus 0pt minus 2pt}
 \setlength{\floatsep}{0pt plus 0pt minus 1pt}
 \setlength{\intextsep}{0pt plus 0pt minus 1pt}
 \setlength{\abovecaptionskip}{0pt}
 \setlength{\belowcaptionskip}{0pt}
}

\usepackage{geometry}
\IEEEoverridecommandlockouts

\usepackage{xcolor} 

\makeatletter

\def\@IEEEstatictitlevskip{{\normalfont\normalsize\vskip 1\baselineskip plus 0.375\baselineskip minus 0.1875\baselineskip}}
\makeatother

\titlesc{Safe Control of Semi-Explicit Differential-Algebraic Systems:\\ Control Barrier Functions with SOS Verification}

\author{Mohamad H. Kazma$^{1}$, Hongchao Zhang$^{2}$, Abdallah A. Albustami$^{1}$, Ben Wooding$^{2}$,\\ Meiyi Ma$^{2}$, Taylor T. Johnson$^{2}$ and Ahmad F. Taha$^{1}$
\thanks{The authors are affiliated with the $^{1}$Department of Civil and Environmental Engineering, and $^{2}$Institute for Convergent Software Integrated Systems at Vanderbilt University, Nashville, TN, USA. This work is partially supported by the National Science Foundation (NSF), United States, under grants 2220401, 2152450, and 2550934, and by the Defense Advanced Research Projects Agency (DARPA) under contract number FA8750-23-C-0518. Any opinions, findings, and conclusions or recommendations expressed in this paper are those of the authors and do not necessarily reflect the views of DARPA or NSF. {\tt\scriptsize \{mohamad.h.kazma, hongchao.zhang, abdallah.b.alalem.albustami, ben.wooding, meiyi.ma, taylor.johnson, ahmad.taha\}@vanderbilt.edu}
}
}

\begin{document}
\newdimen\origiwspc%
\newdimen\origiwstr%
\origiwspc=\fontdimen2\font
\origiwstr=\fontdimen3\font

\maketitle
\thispagestyle{headings} 
\pagestyle{headings}

\markboth{To Appear in the 65$^{\text{th}}$ IEEE Conference on Decision and Control (CDC'2026), Honolulu, Hawaii, December 2026}{}
\begin{abstract}
Differential-algebraic equation (DAE) systems arise in power networks, multibody mechanics, and chemical processes, where algebraic constraints couple dynamic and algebraic states. Standard control barrier function (CBF) methods, designed for ordinary differential equations, fail for DAE systems because the algebraic constraints impose compatibility requirements that may render safety filters infeasible. This paper develops a safety-critical control framework for semi-explicit DAE systems using CBFs that consider compatibility with the algebraic constraints, termed DAE-aware CBFs. We construct projected vector fields on the constraint manifold to account for the hidden dynamics, while ensuring compatibility conditions that keep the algebraic constraints satisfied. We then present sum of squares (SOS) verification conditions that certify both correctness and feasibility of the resulting DAE-aware CBF. The framework is validated on power system case studies.
\end{abstract}
\vspace{-0.5cm}
\section{Introduction}\label{sec:intro}
\IEEEPARstart{S}{afety} of dynamical systems can be formulated as the positive invariance of a safe set~\cite{blanchini1999set}; it requires all trajectories starting within a prescribed region to remain there for all time. For systems modeled by ordinary differential equations (ODEs), \textit{control barrier functions} (CBFs) are considered a computationally efficient tool for realtime safe control~\cite{ames2014control,ames2019control}. CBFs synthesize safety filters that minimally modify a nominal control input to guarantee safety. A myriad of studies have extended CBFs to high-order systems~\cite{xiao2021high}, neural network representations~\cite{dawson2023safe}, systems evolving on manifolds~\cite{Wu2015}, and partially feedback linearizable systems~\cite{Cohen2024}.

While ODE models are common, many engineering systems are described by \emph{differential-algebraic equations} (DAEs)~\cite{Volker2005} that couple dynamic and algebraic states through algebraic constraints. Algebraic constraints arise from physical conservation laws (such as conservation of mass, energy, and power balance) that must be satisfied at all times, and they restrict the allowable state trajectories to a constraint manifold. For DAEs, standard CBF methods are not directly applicable since the safety filter may demand an input that violates the algebraic constraints, rendering the filter infeasible; see Fig.~\ref{fig:geometric}. Safety based on CBFs assumes the existence of a feasible control input satisfying the barrier condition, that can be violated when the algebraic constraint manifold $\mc{M}=\{\x:\phi(\x)=0\}$ imposes compatibility requirements that standard CBF conditions do not account for.

In this paper, we focus on safety-critical control and verification for power system models that admit a regular semi-explicit DAE representation. Regularity ensures that the algebraic states are uniquely determined by the dynamic states, a standard condition in power system DAEs~\cite{Grob2016}.

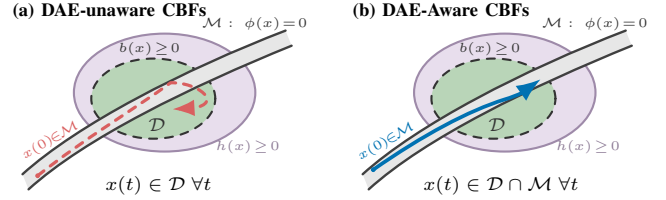
\begin{figure}[t]
\centering
\definecolor{cbBlue}{RGB}{0,114,178}
\definecolor{cbRed}{RGB}{217,95,95}
\definecolor{cbGray}{RGB}{122,122,122}
\definecolor{cbGreen}{RGB}{155,203,151}
\definecolor{cbPurple}{RGB}{198,173,211}
\definecolor{cbBoundary}{RGB}{51,51,51}
\begin{tikzpicture}[>=Latex, line cap=round, line join=round, font=\scriptsize]

\def\W{3.90}
\def\H{2.70}
\def\XS{4.50}

\begin{scope}[xshift=0cm]

    \node[anchor=north west, font=\scriptsize\bfseries, inner sep=2.4pt] at (0,\H) {(a) DAE-unaware CBFs};

    \fill[cbPurple!40] (2.10,1.42) ellipse (1.20 and 0.78);
    \draw[cbPurple!80!black, line width=0.8pt] (2.10,1.42) ellipse (1.20 and 0.78);
    \node[cbPurple!70!black, font=\tiny] at (3.20,0.70) {$h(x)\!\geq\! 0$};

    \draw[cbBoundary, dashed, fill=cbGreen!85, fill opacity=0.75, line width=0.8pt]
      (1.95,1.34) ellipse (0.82 and 0.53);
    \node[cbBoundary] at (2,1) {$\mc{D}$};
    \node[cbBoundary, font=\tiny] at (1.9,2) {$b(x)\!\geq\! 0$};

    \fill[cbGray!18]
      (0.22,0.34) .. controls (0.95,1.00) and (2.35,1.72) .. (3.70,2.24)
      -- (3.78,2.04) .. controls (2.45,1.50) and (1.05,0.82) .. (0.30,0.14) -- cycle;
    \draw[black!75, line width=0.9pt]
      (0.22,0.34) .. controls (0.95,1.00) and (2.35,1.72) .. (3.70,2.24);
    \draw[black!75, line width=0.9pt]
      (0.30,0.14) .. controls (1.05,0.82) and (2.45,1.50) .. (3.78,2.04);
    \node[black!80, font=\tiny, anchor=south west, inner sep=1pt] at (2.55,2.20) {$\mc{M}:\;\phi(x)\!=\!0$};

    \coordinate (x0a) at (0.42,0.32);
    \fill[cbRed] (x0a) circle (0.9pt);
    \node[cbRed, font=\tiny, anchor=west, rotate=30] at ($(x0a)+(-0.3,0.2)$) {$x(0)\hspace{-0.1cm}\in\hspace{-0.1cm}\mc{M}$};
    \draw[->, cbRed, line width=1.2pt, dashed]
    (x0a) .. controls (1.05,0.82) and (1.70,1.20) .. (2.18,1.55)
    .. controls (2.60,1.56) and (3.00,1.18) .. (2.2,1.2);

    \node[anchor=south, font=\scriptsize] at (2.02,-0.05) {$x(t)\in\mc{D}\;\forall t$};

\end{scope}

\begin{scope}[xshift=\XS cm]

    \node[anchor=north west, font=\scriptsize\bfseries, inner sep=2.4pt] at (0,\H) {(b) DAE-Aware CBFs};

    \fill[cbPurple!40] (2.10,1.42) ellipse (1.20 and 0.78);
    \draw[cbPurple!80!black, line width=0.8pt] (2.10,1.42) ellipse (1.20 and 0.78);
    \node[cbPurple!70!black, font=\tiny] at (3.20,0.70) {$h(x)\!\geq\! 0$};

    \draw[cbBoundary, dashed, fill=cbGreen!85, fill opacity=0.75, line width=0.8pt]
      (1.95,1.34) ellipse (0.82 and 0.53);
    \node[cbBoundary] at (2,1) {$\mc{D}$};
    \node[cbBoundary, font=\tiny] at (1.9,2) {$b(x)\!\geq\! 0$};

    \fill[cbGray!18]
      (0.22,0.34) .. controls (0.95,1.00) and (2.35,1.72) .. (3.70,2.24)
      -- (3.78,2.04) .. controls (2.45,1.50) and (1.05,0.82) .. (0.30,0.14) -- cycle;
    \draw[black!75, line width=0.9pt]
      (0.22,0.34) .. controls (0.95,1.00) and (2.35,1.72) .. (3.70,2.24);
    \draw[black!75, line width=0.9pt]
      (0.30,0.14) .. controls (1.05,0.82) and (2.45,1.50) .. (3.78,2.04);
    \node[black!80, font=\tiny, anchor=south west, inner sep=1pt] at (2.55,2.20) {$\mc{M}:\;\phi(x)\!=\!0$};

    \coordinate (x0b) at (0.36,0.40);
    \fill[cbBlue] (x0b) circle (0.9pt);
    \node[cbBlue, font=\tiny, anchor=west, rotate=30] at ($(x0b)+(-0.3,0.1)$) {$x(0)\hspace{-0.1cm}\in\hspace{-0.1cm}\mc{M}$};
    \draw[->, cbBlue, line width=1.4pt]
      (x0b) .. controls (1.02,0.86) and (1.62,1.18) .. (2.10,1.33)
      .. controls (2.30,1.39) and (2.46,1.47) .. (2.58,1.60);

    \node[anchor=south, font=\scriptsize] at (2.05,-0.05) {$x(t)\in\mc{D}\cap\mc{M}\;\forall t$};

\end{scope}

\end{tikzpicture}
\caption{(a)~A standard CBF safety filter applied to a DAE system may produce an input whose resulting velocity is not tangent to the constraint manifold $\mc{M}$, rendering the trajectory infeasible. (b)~The proposed DAE-aware CBF produces an input that ensures $x(t)\in\mc{D}\cap\mc{M}\;\forall t\geq 0$.}
\label{fig:geometric}
\end{figure}

\parstart{Related Work} Existing approaches for DAE safety verification include reachability analysis based on zonotopes~\cite{Althoff2014}, Hamilton-Jacobi level set methods~\cite{Cross2008}, and hybrid systems reachability techniques~\cite{Dang2004}. In power systems, compositional reachability methods have been applied to transient stability analysis~\cite{Althoff2014b} and reachability analysis has been used to characterize microgrid dynamics under large disturbances~\cite{Zhou2021}, while barrier certificates via SOS programming have been used for voltage constraint verification~\cite{Pedersen2016} and distributed safe operation of inverter-based microgrids~\cite{Kundu2019}. These methods provide offline verification or certification, yet do not construct realtime safety filters. On the control synthesis side, DAE stabilization methods based on polynomial Lyapunov functions~\cite{Kundu2015}, dissipativity conditions~\cite{Jensen2025}, and LMI-based rational control~\cite{Portela2025} target asymptotic stability rather than strict set invariance. CBFs have been applied to microgrid voltage regulation~\cite{Michos2025}; however, this formulation does not account for algebraic constraints arising from DAE structure. Recently, CBFs for DAEs have been introduced in~\cite{Park2024b} for flexible link manipulators; however, that work does not explicitly derive how control inputs can ensure algebraic constraint manifold invariance or present a method to verify candidate CBFs. A general framework for DAE-aware CBFs that addresses both safety filter design and formal verification remains an open problem.

\parstart{Paper Contributions} This paper introduces a safety-critical control framework for semi-explicit DAE systems by considering DAE-aware CBFs. We address two fundamental challenges. \textit{(i)} We ensure that the control input maintains invariance of the constraint manifold while enforcing safe set invariance; and \textit{(ii)} we formally verify that the resulting safety filter is both correct and feasible. The contributions are as follows.
\vspace{-0.15cm}
\begin{itemize}
    \item We construct projected vector fields on the constraint manifold that encode the hidden coupling between dynamic and algebraic states. Under this representation, we derive DAE-aware CBF conditions that simultaneously enforce set invariance and manifold compatibility on $\mc{M}$; these conditions are then extended to high-order CBFs for barrier functions of arbitrary relative degree.
    \item We develop SOS-based verification conditions for the resulting CBF that establish correctness (the barrier set lies within the safe set on the manifold) and feasibility (the compatible control input exists and satisfies the safety and algebraic constraints). Both conditions reduce to semidefinite programs.
\end{itemize}
\vspace{-0.15cm}

The work herein considers power system models represented as semi-explicit DAE systems, coupling generator dynamics with algebraic power flow constraints. The extended version~\cite{Zhang2026} generalizes the framework to higher-index DAEs, develops neural network CBF synthesis with SMT-based verification, and presents mechanical system benchmarks. The present paper, however, focuses on a power system case study and introduces an autonomous forward invariance result not covered in~\cite{Zhang2026}.


\parstart{Notation} The boundary and interior of a set $\mc{S}$ are denoted $\partial\mc{S}$ and $\mr{int}(\mc{S})$, respectively. A continuous function $\alpha:\R\to\R$ is \emph{class-$\mc{K}_\infty$} if $\alpha(0)=0$, $\alpha$ is strictly increasing, and $\alpha(r)\to\infty$ as $r\to\infty$. The gradient of a smooth function $y:\Rn{n}\to\R$ is denoted $\nabla y(\x)\in\Rn{n}$, and $\nabla_{\x_i} y(\x)$ denotes the partial gradient with respect to $\x_i$. For a smooth vector field $f:\Rn{n}\to\Rn{n}$, the Lie derivative of $y$ along $f$ is $\mc{L}_f y(\x):=\nabla y(\x)^\top f(\x)$, with higher-order Lie derivatives $\mc{L}_f^k y := \mc{L}_f(\mc{L}_f^{k-1} y)$ and $\mc{L}_f^0 y := y$.

\parstart{Paper Organization} Section~\ref{sec:prelim} provides the required preliminaries and presents the problem formulation. Section~\ref{sec:dae_cbf} introduces DAE-aware CBFs and derives safety conditions for semi-explicit DAE systems. Section~\ref{sec:verification} introduces the SOS-based verification method. Section~\ref{sec:casestudy} presents power system case studies. The paper is concluded in Section~\ref{sec:conclusion}.

\section{Preliminaries and Problem Formulation}\label{sec:prelim}

\subsection{System Model}\label{subsec:model}
We consider a nonlinear control affine semi-explicit DAE system in continuous-time given as follows.
\begin{subequations}\label{eq:dyn}
\begin{align}
\dot{\x}_d(t) &= f_d(\x(t)) + g_d(\x(t))\,\ci(t), \label{eq:dyn_diff}\\
0 &= \phi(\x(t)), \label{eq:dyn_alg}
\end{align}
\end{subequations}
where $\x_d(t)\in\Rn{n_d}$ denotes the dynamic states, $\x_a(t)\in\Rn{n_a}$ denotes the algebraic states, $\x(t)=[\x_d(t)^\top,\x_a(t)^\top]^\top\in\Rn{n_x}$ is the full state with $n_x=n_d+n_a$, and $\ci(t)\in\mc{U}\subseteq\Rn{n_u}$ is the control input, where $\mc{U}$ denotes the admissible input set. The smooth maps $f_d:\Rn{n_x}\to\Rn{n_d}$ and $g_d:\Rn{n_x}\to\Rn{n_d\times n_u}$ are the drift and input fields, and $\phi:\Rn{n_x}\to\Rn{n_m}$ is the algebraic constraint map. The constraint $\phi(\x)=0$ restricts all trajectories to the \emph{constraint manifold} $\mc{M}:=\{\x\in\Rn{n_x}:\phi(\x)=0\}$. We note that $\phi(\x)$ is not directly controlled, that is, $\phi$ does not depend on $\ci$. This is the case in power system DAEs, where algebraic power flow balance equations couple with control inputs acting on the generator dynamics. The following definition is essential for characterizing the DAE structure.

\begin{definition}[Differentiation index~\cite{Volker2005}]\label{def:index}
The differentiation index $\nu$ of~\eqref{eq:dyn} is the minimum number of differentiations of the algebraic constraint required to obtain an explicit ODE via algebraic manipulations.
\end{definition}

Throughout this paper, we consider index-1 systems~\eqref{eq:dyn}, for which a single differentiation of $\phi(\x)=0$ suffices to determine $\dot{\x}_a$. Differentiating~\eqref{eq:dyn_alg} requires the evaluation of the Jacobians of $\phi$ with respect to the dynamic $\x_d$ and algebraic $\x_a$ states, defined as follows.
\begin{equation}\label{eq:jacobians}
J_d(\x) := \nabla_{\x_d}\phi(\x), \qquad J_a(\x) := \nabla_{\x_a}\phi(\x).
\end{equation}

\begin{assumption}[Regularity]\label{assumption:regularity}
For all $\x\in\mc{C}\cap\mc{M}$, the algebraic Jacobian $J_a(\x)\in\Rn{n_m\times n_a}$ is square and nonsingular, i.e., $n_m=n_a$ and $\det(J_a(\x))\neq 0$.
\end{assumption}

Assumption~\ref{assumption:regularity} is the standard regularity condition for index-1 DAEs~\cite{Volker2005,Grob2016}; it guarantees that, given $\dot{\x}_d$, the algebraic state derivative $\dot{\x}_a$ is uniquely determined. We note that the assumption is a mild condition and is satisfied by semi-explicit power system DAE models~\cite{Grob2016}.

\subsection{Safety and Control Barrier Functions}\label{subsec:cbf_background}
Safety of dynamical systems refers to the property that the system state remains within a prescribed set for all time. Let $h:\Rn{n_x}\to\R$ define the safe set $\mc{C}:=\{\x\in\Rn{n_x}:h(\x)\geq 0\}$.
\begin{definition}[Safety]\label{def:positive_invariance}
A set $\mc{C}\subseteq\Rn{n_x}$ is \emph{positively invariant} under~\eqref{eq:dyn} and policy $\mu$ if $\x(0)\in\mc{C}$ and $\ci(t)=\mu(\x(t))$ for all $t\geq 0$ implies $\x(t)\in\mc{C}$ for all $t\geq 0$.
\end{definition}

For ODE systems ($\phi(\x)\equiv 0$), a CBF $b:\Rn{n_x}\to\R$ defines the barrier set $\mc{D}:=\{\x\in\Rn{n_x}:b(\x)\geq 0\}\subseteq\mc{C}$. The following result guarantees positive invariance of $\mc{D}$ under appropriate control.

\begin{theorem}[\hspace{-0.01cm}\cite{ames2019control}]\label{theorem:CBF_safety}
Suppose $b(\x)$ is a CBF with $b(\x(0))\geq 0$ and the input $\ci(t)$ satisfies
\begin{equation}\label{eq:cbf_ode}
\mc{L}_{f} b(\x) + \mc{L}_{g} b(\x)\,\ci \geq -\alpha\big(b(\x)\big),
\end{equation}
for all $t\geq 0$, where $f:=f_d$ and $g:=g_d$. Then $\mc{D}$ is positively invariant.
\end{theorem}

In practice, the input is obtained by solving a CBF-QP that minimally modifies a nominal controller $\ci_{\mr{nom}}$ subject to~\eqref{eq:cbf_ode}
\begin{equation}\label{eq:cbf_qp_ode}
\min_{\ci\in\mc{U}} \;\normnsz{\ci-\ci_{\mr{nom}}}^2 \quad \text{s.t.} \quad \eqref{eq:cbf_ode}.
\end{equation}
Theorem~\ref{theorem:CBF_safety} assumes the existence of a control input satisfying~\eqref{eq:cbf_ode}. A CBF $b(\x)$ is \emph{valid} if \textit{(i)}~$\mc{D}\subseteq\mc{C}$ (correctness); and \textit{(ii)}~for every $\x\in\mc{D}$, there exists $\ci\in\mc{U}$ satisfying~\eqref{eq:cbf_ode} (feasibility). Correctness ensures that the barrier set is contained within the safe set, while feasibility guarantees that the CBF-QP~\eqref{eq:cbf_qp_ode} admits a solution at every state.

\subsection{Problem Formulation}\label{subsec:problem}
For DAE systems~\eqref{eq:dyn}, trajectories are restricted to $\mc{M}$, thus safety must be enforced on $\mc{D}\cap\mc{M}$ rather than $\mc{D}$ alone. Directly applying~\eqref{eq:cbf_ode} to a DAE introduces two challenges: \textit{(i)}~a control input satisfying the barrier inequality may violate $\phi(\x)=0$, since the Lie derivatives do not account for the coupling between $\x_d$ and $\x_a$ imposed by~\eqref{eq:dyn_alg}; and \textit{(ii)}~verifying that a compatible, safe input exists at every state in $\mc{D}\cap\mc{M}$ requires certifying feasibility of the barrier condition subject to the algebraic constraints. Without accounting for these, the CBF-QP~\eqref{eq:cbf_qp_ode} may produce an input that is incompatible with $\mc{M}$ (Fig.~\ref{fig:geometric}), rendering it infeasible. To overcome these challenges, we solve the following problem.
\begin{problem}\label{problem:dae_cbf}
Given a semi-explicit DAE~\eqref{eq:dyn} satisfying Assumption~\ref{assumption:regularity}, a safe set $\mc{C}$, and a candidate barrier function $b(\x)$: \textit{(i)}~design a control policy $\mu:\Rn{n_x}\to\Rn{n_u}$ such that $\mc{D}\cap\mc{M}$ is positively invariant for all $\x(0)\in\mc{D}\cap\mc{M}$; and \textit{(ii)}~verify that the resulting safety filter is both correct ($\mc{D}\cap\mc{M}\subseteq\mc{C}\cap\mc{M}$) and feasible (a compatible, safe input exists at every state).
\end{problem}

To address the above problem, we construct projected vector fields on $\mc{M}$ and derive CBF conditions that jointly enforce manifold compatibility and set invariance (Section~\ref{sec:dae_cbf}). Then we develop SOS-based verification methods to certify correctness and feasibility of the resulting CBFs (Section~\ref{sec:verification}); see Fig.~\ref{fig:framework} for an overview.

\section{DAE-Aware CBFs for Semi-Explicit DAEs}\label{sec:dae_cbf}

This section presents the main theoretical results: \textit{(i)}~projected vector fields that reveal hidden constraints imposed by the algebraic equations; \textit{(ii)}~a CBF condition that enforces safety on the constraint manifold; and \textit{(iii)}~a high-order extension for barrier functions with arbitrary relative degree.
\subsection{Projected Vector Fields and DAE-Aware CBFs}\label{subsec:projected_fields}
For any feasible trajectory of~\eqref{eq:dyn}, the algebraic constraints must hold for all time, that is, $\phi(\x(t))=0$ for all $t\geq 0$. Differentiating with respect to time and substituting $\dot{\x}_d=f_d(\x)+g_d(\x)\ci$ yields the \emph{compatibility condition}
\begin{equation}\label{eq:compatibility}
J_d(\x)\big(f_d(\x)+g_d(\x)\ci\big) + J_a(\x)\dot{\x}_a = 0.
\end{equation}
Under Assumption~\ref{assumption:regularity}, $J_a(\x)$ is nonsingular and~\eqref{eq:compatibility} admits the unique solution
\begin{equation}\label{eq:xa_dot}
\dot{\x}_a(\x,\ci) = -J_a(\x)^\dagger J_d(\x)\big(f_d(\x)+g_d(\x)\ci\big),
\end{equation}
where $J_a(\x)^\dagger$ denotes the Moore-Penrose pseudoinverse. For regular semi-explicit DAEs with $n_m=n_a$, $J_a(\x)$ is square and invertible under Assumption~\ref{assumption:regularity}, thus $J_a(\x)^\dagger=J_a(\x)^{-1}$. Now, augmenting the differential dynamics with~\eqref{eq:xa_dot} gives the \emph{projected vector fields}
\begin{equation}\label{eq:projected_fields}
\hat{f}(\x) \!:=\! \begin{bmatrix} f_d(\x) \\ -J_a^\dagger J_d f_d \end{bmatrix}\!,\;\;
\hat{g}(\x) \!:=\! \begin{bmatrix} g_d(\x) \\ -J_a^\dagger J_d g_d \end{bmatrix}\!,
\end{equation}
where the argument $(\x)$ is suppressed for brevity.
Since $\nabla\phi = [J_d\;\; J_a]$, expanding gives $\nabla\phi\,\hat{f} = J_d f_d - J_a J_a^\dagger J_d f_d$; under Assumption~\ref{assumption:regularity}, $J_a$ has full row rank so $J_a J_a^\dagger = I_{n_a}$, yielding $\nabla\phi(\x)\hat{f}(\x)=0$ and, similarly, $\nabla\phi(\x)\hat{g}(\x)=0$.
Thus, any trajectory of $\dot{\x}=\hat{f}(\x)+\hat{g}(\x)\ci$ starting on $\mc{M}$ remains on $\mc{M}$.

Based on the above, we define the Lie derivatives along the projected fields as $\mc{L}_{\hat{f}}b(\x):=\nabla b(\x)^\top\hat{f}(\x)$ and $\mc{L}_{\hat{g}}b(\x):=\nabla b(\x)^\top\hat{g}(\x)$, and introduce the \emph{projection operator}
\begin{equation}\label{eq:projection_operator}
P(\x) := I_{n_m} - J_a(\x)J_a(\x)^\dagger.
\end{equation}
Thus, the compatibility condition~\eqref{eq:compatibility} can be equivalently written as $P(\x)\big(J_d(\x)f_d(\x)+J_d(\x)g_d(\x)\ci\big)=0$.

\begin{remark}\label{remark:reduction_to_ode}
For regular index-1 DAEs with $n_m=n_a$, $J_a(\x)$ is square and invertible on $\mc{C}\cap\mc{M}$, so $J_a(\x)^\dagger=J_a(\x)^{-1}$ and $P(\x)=I_{n_m}-J_a J_a^{-1}=0$. The compatibility condition then drops from the DAE-aware CBF-QP. However, the CBF inequality itself \emph{still} uses the projected fields $(\hat{f},\hat{g})$ from~\eqref{eq:projected_fields}, which differ from the original $(f,g)$. For~\eqref{eq:projected_fields}, the lower blocks $-J_a^{-1}J_d f_d$ and $-J_a^{-1}J_d g_d$ encode how the algebraic states respond to the differential dynamics. A naive CBF evaluated along $(f,g)$ ignores this coupling and can drive trajectories off $\mc{M}$. The consistency with algebraic constraints, thus, is enforced by both the projected fields~\eqref{eq:projected_fields} and the compatibility constraint~\eqref{eq:compatibility}.
\end{remark}

Under Assumption~\ref{assumption:regularity} and Remark~\ref{remark:reduction_to_ode}, the projected dynamics $\dot{\x}=\hat{f}(\x)+\hat{g}(\x)\ci$ are well-defined and locally Lipschitz on $\mc{C}\cap\mc{M}$, ensuring that standard ODE existence and uniqueness results apply to the system on $\mc{M}$. The following theorem provides sufficient conditions for the positive invariance of $\mc{D}\cap\mc{M}$ under the projected dynamics.

\begin{figure}[t]
\centering
\definecolor{fink}{RGB}{18,18,18}
\definecolor{fg1}{RGB}{248,248,248}
\definecolor{fg3}{RGB}{220,220,220}
\definecolor{fg5}{RGB}{85,85,85}
\definecolor{faccentblue}{RGB}{39,76,119}
\resizebox{0.9\columnwidth}{!}{%
\begin{tikzpicture}[
    font=\rmfamily\fontsize{8.5}{10.5}\selectfont,
    >=Latex, line cap=round, line join=round,
    block/.style={rectangle, draw=fink, line width=0.70pt, rounded corners=2.5pt,
      fill=fg1, minimum height=1.10cm, minimum width=2.30cm, text=fink,
      align=center, inner sep=5pt},
    result/.style={block, fill=faccentblue!12, draw=faccentblue!70!black},
    arr/.style={->, line width=0.80pt, fink},
    arrlabel/.style={midway, text=fg5, font=\rmfamily\fontsize{7.5}{9.5}\selectfont, inner sep=2pt}
]
\node[block] (dae) at (0,0)
  {Semi-explicit\\DAE system\\{\fontsize{7.5}{9}\selectfont $\dot{x}_d\!=\!f_d(x)\!+\!g_d(x)u$}\\[-1pt]{\fontsize{7.5}{9}\selectfont $0=\phi(x)$}};
\node[block, right=1.30cm of dae] (proj)
  {Projected\\vector fields\\{\fontsize{7.5}{9}\selectfont $\hat{f}(x),\;\hat{g}(x)$}};
\node[block, right=1.30cm of proj] (cbf)
  {DAE-aware\\CBF conditions\\{\fontsize{7.5}{9}\selectfont $P(x)(J_d f_d\!+\!J_d g_d u)\!=\!0$}\\[-1pt]{\fontsize{7.5}{9}\selectfont $\mathcal{L}_{\hat{f}} b + \mathcal{L}_{\hat{g}} b\,u \!\ge\! {-}\alpha(b)$}};
\node[result, below=0.75cm of proj] (sos)
  {SOS\\verification\\{\fontsize{7.5}{9}\selectfont Positivstellensatz}};
\node[result, right=1.30cm of sos] (ctrl)
  {Safety filter\\{\fontsize{7.5}{9}\selectfont CBF-QP on $\mathcal{M}$}};
\draw[arr] (dae) -- (proj) node[arrlabel, above] {};
\draw[arr] (proj) -- (cbf) node[arrlabel, above] {Theorem~\ref{thm:positive_invariance}};
\draw[arr] (cbf) -- (ctrl) node[arrlabel, right, pos=0.45] {};
\draw[arr] (cbf) -- (sos) node[arrlabel, left, pos=0.5,xshift=-5pt, yshift=0pt] {Theorem~\ref{thm:verification}};
\draw[arr] (sos) -- (ctrl) node[arrlabel, below] {certificates};
\end{tikzpicture}%
}
\vspace{-0.05cm}
\caption{Overview of the proposed framework. Projected vector fields~\eqref{eq:projected_fields} capture the coupling between dynamic and algebraic states on $\mc{M}$. The DAE-aware CBF condition (Theorem~\ref{thm:positive_invariance}) and its high-order extension enforce safety on the constraint manifold, while SOS verification (Theorem~\ref{thm:verification}) certifies correctness and feasibility of the resulting safety filter.}
\label{fig:framework}
\end{figure}
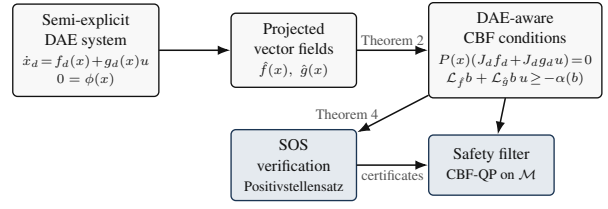

\begin{theorem}[Safety of semi-explicit DAEs]\label{thm:positive_invariance}
Consider the semi-explicit DAE~\eqref{eq:dyn} with $\nu=1$ satisfying Assumption~\ref{assumption:regularity}. Suppose $b(\x(0))\geq 0$ and $\x(0)\in\mc{D}\cap\mc{M}$. The set $\mc{D}\cap\mc{M}$ is positively invariant if the input $\ci(t)$ satisfies, for all $\x\in\mc{D}\cap\mc{M}$,
\begin{subequations}\label{eq:invariance_conditions}
\begin{align}
P(\x)\big(J_d(\x)f_d(\x) + J_d(\x)g_d(\x)\ci\big) &= 0, \label{eq:compat_cond}\\
\mc{L}_{\hat{f}}b(\x) + \mc{L}_{\hat{g}}b(\x)\ci &\geq -\alpha(b(\x)). \label{eq:cbf_cond}
\end{align}
\end{subequations}
\end{theorem}
\begin{proof}
Since $\phi$ does not depend on the control input, differentiating $\phi(\x(t))=0$ and substituting $\dot{\x}_d=f_d(\x)+g_d(\x)\ci$ yields the compatibility condition~\eqref{eq:compatibility}. Thus, condition~\eqref{eq:compat_cond} is equivalent to~\eqref{eq:compatibility}. Now, for $\nu=1$, $J_a(\x)$ is nonsingular on $\mc{C}\cap\mc{M}$ (Assumption~\ref{assumption:regularity}), thus $\dot{\x}_a$ is uniquely determined by~\eqref{eq:xa_dot} and the projected fields~\eqref{eq:projected_fields} are well-defined with $\nabla\phi(\x)\hat{f}(\x)=0$ and $\nabla\phi(\x)\hat{g}(\x)=0$. Hence any trajectory of $\dot{\x}=\hat{f}(\x)+\hat{g}(\x)\ci$ starting on $\mc{M}$ remains on $\mc{M}$, and differentiating $b$ along the projected dynamics gives $\dot{b}(\x(t))=\mc{L}_{\hat{f}}b(\x)+\mc{L}_{\hat{g}}b(\x)\ci$.
Thus,~\eqref{eq:cbf_cond} yields $\dot{b}\geq -\alpha(b)$, where for $\alpha\in\mc{K}_\infty$ with $\alpha(0)=0$ and $b(\x(0))\geq 0$, the comparison lemma~\cite[Lemma~3.4]{khalil2002nonlinear} gives $b(\x(t))\geq 0$ for all $t\geq 0$. Based on that, any trajectory starting in $\mc{D}\cap\mc{M}$ remains in $\mc{D}\cap\mc{M}$.
\end{proof}

Theorem~\ref{thm:positive_invariance} establishes that condition~\eqref{eq:cbf_cond} is exactly~\eqref{eq:cbf_ode} with $(f,g)$ replaced by $(\hat{f},\hat{g})$, while condition~\eqref{eq:compat_cond} ensures trajectories remain on $\mc{M}$. We note that the Lie derivatives $\mc{L}_{\hat{f}}b$ and $\mc{L}_{\hat{g}}b$ account for the response of the algebraic states through the projected fields~\eqref{eq:projected_fields}; evaluating the CBF condition along the original $(f_d,g_d)$ ignores this coupling and can produce inputs that violate the algebraic constraints. The following corollary is a special case of Theorem~\ref{thm:positive_invariance} for autonomous DAEs. The compatibility condition reduces to a constraint on the drift field $f_d$.

\begin{corollary}[Autonomous DAEs]\label{cor:autonomous_index1}
Consider the DAE system~\eqref{eq:dyn} satisfying Assumption~\ref{assumption:regularity}, with $g_d(\x)\equiv 0$ (no control input). Suppose $b(\x(0))\geq 0$ and $\x(0)\in\mc{D}\cap\mc{M}$. Then $\mc{D}\cap\mc{M}$ is positively invariant if, for all $\x\in\mc{D}\cap\mc{M}$,
\begin{subequations}\label{eq:autonomous_invariance_conditions}
\begin{align}
P(\x)J_d(\x)f_d(\x) &= 0, \label{eq:autonomous_compat}\\
\mc{L}_{\hat{f}}b(\x) &\geq -\alpha(b(\x)). \label{eq:autonomous_cbf}
\end{align}
\end{subequations}
\end{corollary}
\begin{proof}
Set $\ci\equiv 0$ in~\eqref{eq:invariance_conditions}. Condition~\eqref{eq:compat_cond} reduces to~\eqref{eq:autonomous_compat}, and condition~\eqref{eq:cbf_cond} reduces to~\eqref{eq:autonomous_cbf}. The proof then follows directly from Theorem~\ref{thm:positive_invariance}.
\end{proof}

Corollary~\ref{cor:autonomous_index1} is applicable to power system DAEs operating under fixed generation dispatch~\cite{Pedersen2016}, where no control input is available and safety must be certified from the drift dynamics alone; see Section~\ref{sec:casestudy}. Furthermore, in cases where a control input is present however the barrier function $b(\x)$ has relative degree greater than one with respect to the projected dynamics, the conditions of Theorem~\ref{thm:positive_invariance} do not directly apply. The following subsection extends the framework to account for the relative degree of $b(\x)$.

\subsection{Relative Degree and High-Order CBFs}\label{subsec:hocbf}
The projected dynamics~\eqref{eq:projected_fields} may require more than one differentiation of $b(\x)$ before the control input appears. To formalize this, we recall the notion of relative degree.

\begin{definition}[Relative degree~\cite{Isidori2013}]\label{def:relative_degree}
The \emph{relative degree} of a smooth function $b(\x)$ with respect to a control affine system $\dot{\x}=f(\x)+g(\x)\ci$ is $d\geq 1$ if $d$ is the smallest integer such that \textit{(i)} $\mc{L}_g\mc{L}_f^k b(\x)=0$ for all $k<d-1$; and \textit{(ii)} $\mc{L}_g\mc{L}_f^{d-1}b(\x)\neq 0$.
\end{definition}

Note that Theorem~\ref{theorem:CBF_safety} is considered when $b(\x)$ has relative degree one, where $\ci$ appears directly in $\dot{b}(\x)$. For the projected dynamics on $\mc{M}$, the relative degree of $b(\x)$ is determined by the structure of $\hat{g}$ in~\eqref{eq:projected_fields}. That is, the control enters through the lower block $-J_a^\dagger J_d g_d$, and this propagation can increase the effective relative degree beyond what would be observed in the unconstrained ODE. The following lemma formalizes this relationship for semi-explicit DAE systems.

\begin{lemma}[Index and relative degree]\label{lem:index_rel_degree}
Consider the DAE system~\eqref{eq:dyn} satisfying Assumption~\ref{assumption:regularity}. Let $d'$ denote the smallest integer such that $J_d(\x)\mc{L}_{g_d}\mc{L}_{f_d}^{d'-1}\x_d(\x)\neq 0$. Then, the relative degree, $d$, of the algebraic constraint $\phi(\x)$ with respect to system~\eqref{eq:dyn} satisfies $d=\nu+d'-1$. For index-1 systems ($\nu=1$), we have $d=d'$.
\end{lemma}
\begin{proof}
For a single differentiation of $\phi(\x(t))=0$, the compatibility condition~\eqref{eq:compatibility} yields $J_d(\x)\dot{\x}_d+J_a(\x)\dot{\x}_a=0$. Since $\nu=1$, $J_a(\x)$ is invertible and $\dot{\x}_a$ is determined by $\dot{\x}_d$. The control $\ci$ appears in the $\nu$-th derivative of $\phi$ through $J_d(\x)$; for $\ci$ to have a nonzero coefficient, $d'$ additional derivatives of $\x_d$ through $J_d(\x)$ are required (Definition~\ref{def:relative_degree}). Hence $d=\nu+d'-1$, and for $\nu=1$, $d=d'$.
\end{proof}

Lemma~\ref{lem:index_rel_degree} reveals that even for index-1 DAEs, the actuation structure of the differential subsystem can produce $d'>1$, requiring a high-order CBF (HOCBF). When $d'=1$ (i.e., $J_d(\x)g_d(\x)\neq 0$), the standard CBF condition of Theorem~\ref{thm:positive_invariance} suffices. When $d'>1$, the control input does not appear in the first $d-1$ time derivatives of $b(\x)$ along $\mc{M}$, and the following extension of Theorem~\ref{thm:positive_invariance} is required.
\begin{corollary}[DAE-aware HOCBF for semi-explicit systems]\label{cor:hocbf}
Consider a semi-explicit DAE system satisfying Assumption~\ref{assumption:regularity}. Let $b(\x)$ have relative degree $d\geq 1$ on $\mc{M}$ with respect to~\eqref{eq:projected_fields}. Define the HOCBF hierarchy $\psi_0(\x):=b(\x)$ and
\begin{equation}\label{eq:hocbf_hierarchy}
\psi_i(\x) := \mc{L}_{\hat{f}}\psi_{i-1}(\x) + \alpha_i(\psi_{i-1}(\x)), \quad i=1,\ldots,d,
\end{equation}
for class-$\mc{K}_\infty$ functions $\alpha_i$. Then $\mc{D}\cap\mc{M}$ is positively invariant if there exists $\ci$ such that
\begin{subequations}\label{eq:hocbf_conditions}
\begin{align}
P(\x)\big(J_d(\x)f_d(\x) + J_d(\x)g_d(\x)\ci\big) &= 0, \label{eq:hocbf_compat}\\
\mc{L}_{\hat{f}}\psi_{d-1}(\x) + \alpha_d(\psi_{d-1}(\x)) + \mc{L}_{\hat{g}}\psi_{d-1}(\x)\ci &\geq 0, \label{eq:hocbf_cbf}
\end{align}
\end{subequations}
hold for all $\x\in\mc{D}\cap\mc{M}$ with $\psi_i(\x(0))\geq 0$ for $i=0,\ldots,d$.
\end{corollary}
\begin{proof}
Condition~\eqref{eq:hocbf_compat} ensures $\x(t)\in\mc{M}$ by Theorem~\ref{thm:positive_invariance}. On $\mc{M}$, the dynamics can be written as $\dot{\x}=\hat{f}(\x)+\hat{g}(\x)\ci$, and $b(\x)$ has relative degree $d$ with respect to this system (Definition~\ref{def:relative_degree}). Condition~\eqref{eq:hocbf_cbf} is the HOCBF condition~\cite{xiao2021high} applied to the projected dynamics, giving $b(\x(t))\geq 0$ via iterated comparison arguments.
\end{proof}

Corollary~\ref{cor:hocbf} enforces manifold invariance through the compatibility condition~\eqref{eq:hocbf_compat} and enforces safety via the HOCBF hierarchy in~\eqref{eq:hocbf_cbf} that is constructed from the projected Lie derivatives $\mc{L}_{\hat{f}}$ and $\mc{L}_{\hat{g}}$. When $d=1$, this reduces to $\psi_0=b(\x)$, and condition~\eqref{eq:hocbf_cbf} recovers the CBF inequality~\eqref{eq:cbf_cond}.

\subsection{DAE-Aware CBF-QP}\label{subsec:cbf_qp}
Based on the safety conditions in Theorem~\ref{thm:positive_invariance}, we formulate a DAE-aware CBF-QP for relative degree one as follows.
\begin{equation}\label{eq:dae_cbf_qp}
\begin{aligned}
\min_{\ci\in\mc{U}} \quad & \normnsz{\ci - \ci_{\mr{nom}}}^2 \\
\text{s.t.} \quad & P(\x)\big(J_d(\x)f_d(\x) + J_d(\x)g_d(\x)\ci\big) = 0, \\
& \mc{L}_{\hat{f}}b(\x) + \mc{L}_{\hat{g}}b(\x)\ci \geq -\alpha(b(\x)).
\end{aligned}
\end{equation}
In contrast to the standard CBF-QP~\eqref{eq:cbf_qp_ode}, the DAE-aware formulation~\eqref{eq:dae_cbf_qp} jointly enforces the compatibility condition and the barrier inequality, ensuring that the resulting input maintains both manifold invariance and safety. Since the DAE system~\eqref{eq:dyn} is control-affine, the nominal controller $\ci_{\mr{nom}}$ can be any state-feedback policy (e.g., LQR, droop, or MPC); the CBF-QP~\eqref{eq:dae_cbf_qp} minimally modifies the supplied nominal input. For barriers with relative degree $d>1$, the second constraint in~\eqref{eq:dae_cbf_qp} is replaced by~\eqref{eq:hocbf_cbf} from Corollary~\ref{cor:hocbf}, while the compatibility constraint is unchanged. Theorem~\ref{thm:positive_invariance} guarantees safety for relative degree one, given a barrier function $b(\x)$ that admits a compatible, safe input $\ci\in\mc{U}$ satisfying~\eqref{eq:invariance_conditions} at every $\x\in\mc{D}\cap\mc{M}$. The following definition formalizes the notion of a valid DAE-aware CBF.
\begin{definition}[Valid DAE-aware CBF]\label{def:valid_dae_cbf}
A smooth function $b:\Rn{n_x}\to\R$ is a \emph{valid DAE-aware CBF} for the safe set $\mc{C}$ under DAE~\eqref{eq:dyn} if \textit{(i)}~$\mc{D}\cap\mc{M}\subseteq\mc{C}\cap\mc{M}$ (correctness); and \textit{(ii)}~for every $\x\in\mc{D}\cap\mc{M}$, there exists $\ci\in\mc{U}$ satisfying~\eqref{eq:invariance_conditions} (feasibility).
\end{definition}

Definition~\ref{def:valid_dae_cbf} extends the standard valid CBF notion~\cite{ames2019control} to DAE systems, where correctness ensures the barrier set lies within the safe set on $\mc{M}$, and feasibility ensures the DAE-aware CBF-QP~\eqref{eq:dae_cbf_qp} admits a solution. The next section introduces the framework to verify both conditions using SOS-based verification methods.

\section{SOS-Based Verification}\label{sec:verification}
We now address the verification of a candidate DAE-aware CBF introduced in Definition~\ref{def:valid_dae_cbf}. The following problem formalizes the verification part of Problem~\ref{problem:dae_cbf}.

\begin{problem}\label{problem:verification}
Given a candidate DAE-aware CBF $b:\Rn{n_x}\to\R$, verify \textit{(i)}~\emph{correctness}, i.e., $\mc{D}\cap\mc{M}\subseteq\mc{C}\cap\mc{M}$; and \textit{(ii)}~\emph{feasibility}, i.e., for every $\x\in\mc{D}\cap\mc{M}$, there exists $\ci\in\mc{U}$ such that~\eqref{eq:invariance_conditions} holds.
\end{problem}

\subsection{Correctness Verification}\label{subsec:correctness}
We begin with correctness verification, which requires certifying that $b(\x)\geq 0$ for all $\x\in\mc{D}\cap\mc{M}$. First, recall the Positivstellensatz (Psatz), a central result in real algebraic geometry for certifying emptiness of semialgebraic sets~\cite{parrilo2003semidefinite}. A polynomial $p(\x)$ is \emph{a sum of squares} (SOS) if $p(\x)=\sum_i(p_i(\x))^2$ for some polynomials $p_i$. Let $\Sigma$ denote the SOS cone, $\mc{N}$ the multiplicative monoid, and $\mbb{I}$ the polynomial ideal generated by a given set of polynomials~\cite{parrilo2003semidefinite}. The following result provides a necessary and sufficient condition for emptiness of semialgebraic sets.

\begin{theorem}[Psatz~\cite{parrilo2003semidefinite}]\label{thm:psatz}
The set $\{\x\in\Rn{n_x}:\Phi_i(\x)\geq 0,\;\chi_j(\x)\neq 0,\;\Gamma_\ell(\x)=0\}$ is empty iff there exist $\Phi\in\Sigma$, $\chi\in\mc{N}$, and $\Gamma\in\mbb{I}$ such that $\Phi+\chi^2+\Gamma=0$.
\end{theorem}

Correctness requires $h(\x)\geq 0$ for all $\x\in\mc{D}\cap\mc{M}$. We encode the violation set using a slack variable $s\neq 0$, where $\mc{S}_{\mr{viol}} := \big\{(\x,s):\phi(\x)=0,\;b(\x)\geq 0,
-h(\x)-s^2\geq 0\big\}.$ The following lemma applies Theorem~\ref{thm:psatz} to certify correctness.

\begin{lemma}[Correctness via Psatz]\label{lem:correctness}
Assume $h$, $b$, and $\phi$ are polynomials and $\mc{D}\cap\mc{M}$ is compact. Then $\mc{D}\cap\mc{M}\subseteq\mc{C}\cap\mc{M}$ if there exist SOS polynomials $\sigma_0,\sigma_1,\sigma_2$, polynomial multipliers $\rho_\ell$ for $\ell=1,\ldots,n_m$, and integer $k\geq 1$ such that
\begin{equation}\label{eq:correctness_identity}
\hspace{-0.2cm}\sigma_0 \hspace{-0.05cm}+ \hspace{-0.05cm}\sigma_1 b(\x) \hspace{-0.05cm}+\hspace{-0.05cm} \sigma_2(-h(\x)-s^2) \hspace{-0.05cm}+ \hspace{-0.05cm}s^{2k} \hspace{-0.05cm}+ \hspace{-0.05cm}\sum_{\ell=1}^{n_m} \rho_\ell\phi_\ell(\x) \hspace{-0.05cm}= \hspace{-0.05cm}0.
\end{equation}
\end{lemma}
\begin{proof}
Identity~\eqref{eq:correctness_identity} certifies $\mc{S}_{\mr{viol}}=\emptyset$ by Theorem~\ref{thm:psatz} applied to the cone generated by $b$ and $-h-s^2$, the monoid generated by $s$, and the ideal generated by $\phi_\ell$. Emptiness of $\mc{S}_{\mr{viol}}$ implies that no $\x\in\mc{D}\cap\mc{M}$ satisfies $h(\x)<0$, i.e., $\mc{D}\cap\mc{M}\subseteq\mc{C}\cap\mc{M}$.
\end{proof}

Condition~\eqref{eq:correctness_identity} is checked via an SOS feasibility program that \textit{(i)}~fixes degree bounds on $\sigma_0,\sigma_1,\sigma_2,\rho_\ell$; \textit{(ii)}~enforces $\sigma_i\in\Sigma$; and \textit{(iii)}~matches polynomial coefficients, reducing to a semidefinite program.
\subsection{Feasibility Verification}\label{subsec:feasibility}
We now verify that a compatible, safe input exists at every state in $\mc{D}\cap\mc{M}$. Let $n_u^c$ denote the number of linear input constraints $A_u\ci\leq r_u$, where $A_u\in\Rn{n_u^c\times n_u}$ and $r_u\in\Rn{n_u^c}$. For a given $\x\in\mc{D}\cap\mc{M}$, the existence of $\ci$ that satisfies compatibility, the CBF condition, and input constraints is equivalent to feasibility of the affine system $A_{\mr{dae}}(\x)\ci\leq r_{\mr{dae}}(\x)$, where
\begin{align}
A_{\mr{dae}}(\x) &:= \begin{bmatrix} -P(\x)J_d(\x)g_d(\x) \\ P(\x)J_d(\x)g_d(\x) \\ -\mc{L}_{\hat{g}}b(\x) \\ A_u \end{bmatrix}, \label{eq:A_dae}\\
r_{\mr{dae}}(\x) &:= \begin{bmatrix} P(\x)J_d(\x)f_d(\x) \\ -P(\x)J_d(\x)f_d(\x) \\ \mc{L}_{\hat{f}}b(\x)+\alpha(b(\x)) \\ r_u \end{bmatrix}. \label{eq:r_dae}
\end{align}
The rows $\pm P(\x)J_d(\x)g_d(\x)$ together enforce the compatibility equality~\eqref{eq:compat_cond} via opposing inequalities; the row $-\mc{L}_{\hat{g}}b(\x)$ encodes the CBF condition~\eqref{eq:cbf_cond}; and $A_u$ imposes input bounds. The following result provides the dual alternative used to certify feasibility.

\begin{lemma}[Farkas' Lemma~\cite{matousek2006understanding}]\label{lem:farkas}
The system $A\ci\leq r$ has a solution iff there is no $\lambda\geq 0$ with $\lambda^\top A=0^\top$ and $\lambda^\top r<0$.
\end{lemma}

By Lemma~\ref{lem:farkas}, $A_{\mr{dae}}(\x)\ci\leq r_{\mr{dae}}(\x)$ is feasible if and only if there is no $\lambda\geq 0$ with $\lambda^\top A_{\mr{dae}}(\x)=0^\top$ and $\lambda^\top r_{\mr{dae}}(\x)<0$.

Separating interior and boundary conditions, we define the interior matrices $A_{\mr{int}}(\x)$, $r_{\mr{int}}(\x)$ (compatibility and input bounds only, omitting the CBF row) and boundary matrices
\begin{equation*}
A_{\mr{bnd}}(\x):=\begin{bmatrix}A_{\mr{int}}(\x)\\-\mc{L}_{\hat{g}}b(\x)\end{bmatrix},\quad
r_{\mr{bnd}}(\x):=\begin{bmatrix}r_{\mr{int}}(\x)\\\mc{L}_{\hat{f}}b(\x)\end{bmatrix}.
\end{equation*}

The separation into interior and boundary is motivated by the tangent cone characterization of positive invariance~\cite{blanchini2008set}. At interior points of $\mc{D}$, any compatible input preserves barrier nonnegativity since $b(\x)>0$; at boundary points $\partial\mc{D}$, the input must additionally satisfy the CBF tangency condition to prevent trajectories from leaving $\mc{D}$.

\begin{theorem}[DAE-aware CBF verification]\label{thm:verification}
Let $b$ define the barrier set $\mc{D}:=\{\x:b(\x)\geq 0\}$. Then $b$ is a valid DAE-aware CBF (Definition~\ref{def:valid_dae_cbf}) if and only if
\begin{enumerate}
\item (Correctness) $\mc{D}\cap\mc{M}\subseteq\mc{C}\cap\mc{M}$;
\item (Interior feasibility) For every $\x\in\mr{int}(\mc{D})\cap\mc{M}$, the dual system $\lambda^\top A_{\mr{int}}(\x)=0^\top$, $\lambda^\top r_{\mr{int}}(\x)<0$ admits no $\lambda\geq 0$;
\item (Boundary tangency) For every $\x\in\partial\mc{D}\cap\mc{M}$, the dual system $\lambda^\top \hspace{-0.05cm}A_{\mr{bnd}}(\x)\hspace{-0.05cm}=\hspace{-0.05cm}0^\top$, \hspace{-0.05cm} $\lambda^\top \hspace{-0.05cm}r_{\mr{bnd}}(\x)\hspace{-0.05cm}<\hspace{-0.05cm}0$ admits no $\hspace{-0.05cm}\lambda\geq \hspace{-0.05cm}0$.
\end{enumerate}
\end{theorem}
\begin{proof}
By Lemma~\ref{lem:farkas}, $A_{\mr{int}}(\x)\ci\leq r_{\mr{int}}(\x)$ is feasible if and only if condition~2 holds, and $A_{\mr{bnd}}(\x)\ci\leq r_{\mr{bnd}}(\x)$ is feasible if and only if condition~3 holds. Together, these guarantee the existence of admissible inputs satisfying compatibility and input bounds at interior points, and additionally the tangency condition at boundary points. Positive invariance of $\mc{D}\cap\mc{M}$ then follows from Theorem~\ref{thm:positive_invariance}.
\end{proof}

We note that the conditions in Theorem~\ref{thm:verification} are exact in that they do not impose a particular choice of class-$\mc{K}_\infty$ function $\alpha$. This provides maximal flexibility in the feasibility analysis, since certifying the existence of \emph{any} compatible safe input suffices. By contrast, Corollary~\ref{cor:hocbf} requires selecting $\alpha_1,\ldots,\alpha_d$ a priori, committing to a specific class-$\mc{K}_\infty$ parameterization in exchange for an online QP formulation.

\subsection{Polynomial SOS Certificates}\label{subsec:sos_certificates}
If the system mapping functions $f_d$, $g_d$, $J_d$, $P$, and $b$, and the sets $\mc{M}$, $\mc{C}$ are polynomial, and $\mc{D}\cap\mc{M}$ is compact, then all entries of $A_{\mr{int}}(\x)$, $A_{\mr{bnd}}(\x)$, $r_{\mr{int}}(\x)$, and $r_{\mr{bnd}}(\x)$ are polynomial in $\x$. Infeasibility of the dual systems can then be certified via the Psatz applied to the semialgebraic sets.

\begin{proposition}[Polynomial DAE-aware certificates]\label{prop:sos_certificates}
Let $f_d$, $g_d$, $J_d$, $P$, and $b$ be polynomial, let $\mc{M}$ and $\mc{C}$ admit polynomial equality and inequality expressions, and let $\mc{D}\cap\mc{M}$ be compact. Then conditions~2 and~3 of Theorem~\ref{thm:verification} hold if there exist SOS polynomials and multipliers that satisfy
\begin{align}
\Phi_{\mr{int}}(\x,\lambda_{\mr{int}}) + s_{\mr{int}}^{2} + \varphi_{\mr{int}}(\x,\lambda_{\mr{int}}) &= 0, \label{eq:sos_int}\\
\Phi_{\mr{bnd}}(\x,\lambda_{\mr{bnd}}) + s_{\mr{bnd}}^{2} + \varphi_{\mr{bnd}}(\x,\lambda_{\mr{bnd}}) &= 0, \label{eq:sos_bnd}
\end{align}
where $\Phi_{\mr{int}}\!\in\!\Sigma_{\mr{int}}$ denotes the cone generated by $b(\x)\!\geq\! 0$, $\lambda_{\mr{int}}\!\geq\! 0$;
$\Phi_{\mr{bnd}}\!\in\!\Sigma_{\mr{bnd}}$ denotes the cone generated by $\lambda_{\mr{bnd}}\!\geq\! 0$; monoid terms encode strict inequalities; and
\begin{align*}
\varphi_{\mr{int}}(\cdot) &= \textstyle\sum_q \rho_q(\cdot)\big[\lambda_{\mr{int}}^\top A_{\mr{int}}(\x)\big]_q + \sum_\ell \rho_\ell^\phi(\cdot)\phi_\ell(\x) \\
& \quad + \rho^s(\cdot)\big(-\lambda_{\mr{int}}^\top r_{\mr{int}}(\x)-s_{\mr{int}}^2\big),\\
\varphi_{\mr{bnd}}(\cdot) &= \textstyle\sum_q \rho_q(\cdot)\big[\lambda_{\mr{bnd}}^\top A_{\mr{bnd}}(\x)\big]_q + \sum_\ell \rho_\ell^\phi(\cdot)\phi_\ell(\x) \\
& \quad + \rho^b(\cdot)b(\x) + \rho^s(\cdot)\big(-\lambda_{\mr{bnd}}^\top r_{\mr{bnd}}(\x)-s_{\mr{bnd}}^2\big).
\end{align*}
\end{proposition}
\begin{proof}
The proof follows from Lemma~\ref{lem:farkas}. The interior and boundary dual systems are infeasible iff their semialgebraic representations are empty. Theorem~\ref{thm:psatz} then certifies such emptiness through the stated cone, monoid, and ideal identities. Given that all functions are polynomial, this is reduced to an SOS feasibility program.
\end{proof}

Proposition~\ref{prop:sos_certificates}, combined with Lemma~\ref{lem:correctness}, provides a complete SOS-based verification framework for DAE-aware CBFs, where correctness is certified by~\eqref{eq:correctness_identity} and feasibility is certified by~\eqref{eq:sos_int}-\eqref{eq:sos_bnd}. Both reduce to semidefinite programs that can be solved using standard SOS solvers. While Proposition~\ref{prop:sos_certificates} requires polynomial mappings, non-polynomial algebraic constraints (e.g., trigonometric power flow equations) can be handled by introducing auxiliary variables that lift the system to a polynomial DAE on an extended manifold~\cite{Kundu2015}; see Section~\ref{sec:casestudy}.

\section{Case Studies}\label{sec:casestudy}
Based on the results in Sections~\ref{sec:dae_cbf} and~\ref{sec:verification}, we investigate the following research questions.
\begin{itemize}
    \item $(\mr{Q}1)$ Does the DAE-aware CBF in Theorem~\ref{thm:positive_invariance} guarantee manifold compatibility and safe set invariance?
    \item $(\mr{Q}2)$ Do the candidate barriers satisfy the validity conditions in Theorem~\ref{thm:verification}, as certified by the SOS programs in Lemma~\ref{lem:correctness} and Proposition~\ref{prop:sos_certificates}?
    \item $(\mr{Q}3)$ For an autonomous power DAE from~\cite{Pedersen2016} with $\ci\equiv 0$, does Corollary~\ref{cor:autonomous_index1} certify positive invariance of $\mc{D}\cap\mc{M}$?
\end{itemize}
To answer these questions, we consider two semi-explicit power system DAEs. We first apply the DAE-aware CBF-QP (Theorem~\ref{thm:positive_invariance}) and the SOS verification framework (Theorem~\ref{thm:verification} via Lemma~\ref{lem:correctness} and Proposition~\ref{prop:sos_certificates}) to a nonlinear network preserving DAE model~\cite{Nugroho2023,Grob2016,Kundu2015}. We then consider a two bus polynomial benchmark from~\cite{Pedersen2016} to verify Corollary~\ref{cor:autonomous_index1}.

\subsection{Nonlinear Power System DAE Model}\label{subsec:nl_benchmark}
The model is a classical network preserving power system DAE with a second-order generator model~\cite{Kundu2015}. This formulation couples the swing equations for generator dynamics with algebraic power balance constraints at load buses. The solvability and index-1 structure of such models are established in~\cite{Grob2016}.
We consider a two bus network with one generator bus and one load bus; see Fig.~\ref{fig:two_bus_circuit}. Let the dynamic states be $\x_d=[\delta_g,\omega_g]^\top\in\Rn{2}$ and the algebraic states be $\x_a=[\delta_\ell,V_\ell]^\top\in\Rn{2}$, where $\delta_g$ and $\omega_g$ denote generator rotor angle and speed, and $\delta_\ell$ and $V_\ell$ denote load bus voltage angle and magnitude. The nonlinear semi-explicit DAE model can be written as
\begin{subequations}\label{eq:nl_power_dae}
\begin{align}
\dot{\delta}_g &= \omega_g,\label{eq:nl_power_dae_delta}\\
M\dot{\omega}_g &= -D\omega_g + P_m - P_e(\x) + \ci,\label{eq:nl_power_dae_omega}\\
0 &= \phi_1(\x):=P_\ell^{\mr{spec}}-P_\ell(\delta_g,\delta_\ell,V_\ell),\label{eq:nl_power_dae_phi1}\\
0 &= \phi_2(\x):=Q_\ell^{\mr{spec}}-Q_\ell(\delta_g,\delta_\ell,V_\ell),\label{eq:nl_power_dae_phi2}
\end{align}
\end{subequations}
where $\vartheta:=\delta_g-\delta_\ell$ is the angle difference across the transmission line, $P_e(\x) := E_g^2G_{gg}+E_gV_\ell(G_{g\ell}\cos\vartheta+B_{g\ell}\sin\vartheta)$ is the electrical power injected by the generator, $P_\ell := V_\ell^2G_{\ell\ell}+V_\ell E_g(G_{\ell g}\cos\vartheta-B_{\ell g}\sin\vartheta)$ and $Q_\ell := -V_\ell^2B_{\ell\ell}-V_\ell E_g(G_{\ell g}\sin\vartheta+B_{\ell g}\cos\vartheta)$ are the active and reactive power absorbed at the load bus. The parameters $M>0$ and $D>0$ are inertia and damping coefficients, $P_m$ is mechanical power, $P_\ell^{\mr{spec}},Q_\ell^{\mr{spec}}$ are active and reactive load setpoints, $E_g$ is the internal generator voltage magnitude, and $G_{ij},B_{ij}$ are entries of the network admittance matrix. 

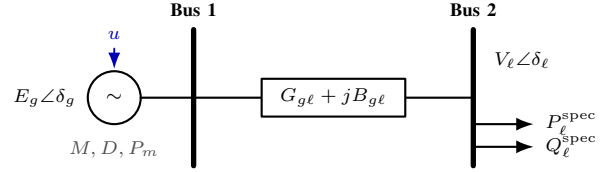
\begin{figure}[t]
\centering
\begin{tikzpicture}[>=Latex, line cap=round, line join=round, font=\scriptsize,
  busbar/.style={line width=2pt},
  pline/.style={line width=0.8pt}]

\draw[busbar] (-0.35,0) -- (-0.35,1.8);
\node[anchor=south, font=\scriptsize\bfseries] at (-0.35,1.85) {Bus~1};

\draw[pline] (-1.4,0.9) circle (0.35);
\node at (-1.4,0.9) {$\sim$};
\draw[pline] (-1.05,0.9) -- (-0.35,0.9);

\node[anchor=east, font=\scriptsize] at (-1.78,0.9) {$E_g\angle\delta_g$};
\node[anchor=north, font=\scriptsize, text=black!70] at (-1.4,0.45) {$M,D,P_m$};

\draw[->, line width=0.7pt, blue!70!black] (-1.4,1.55) -- (-1.4,1.28);
\node[anchor=south, font=\scriptsize, blue!70!black] at (-1.4,1.55) {$\ci$};

\draw[pline] (-0.35,0.9) -- (3.35,0.9);

\draw[pline, fill=white] (0.55,0.65) rectangle (2.45,1.15);
\node at (1.50,0.9) {$G_{g\ell}+jB_{g\ell}$};

\draw[busbar] (3.35,0) -- (3.35,1.8);
\node[anchor=south, font=\scriptsize\bfseries] at (3.35,1.85) {Bus~2};

\node[anchor=west, font=\scriptsize] at (3.50,1.4) {$V_\ell\angle\delta_\ell$};

\draw[->, pline] (3.35,0.55) -- (4.15,0.55);
\draw[->, pline] (3.35,0.25) -- (4.15,0.25);
\node[anchor=west, font=\scriptsize] at (4.18,0.55) {$P_\ell^{\mr{spec}}$};
\node[anchor=west, font=\scriptsize] at (4.18,0.25) {$Q_\ell^{\mr{spec}}$};

\end{tikzpicture}
\vspace{-0.1cm}
\caption{Single-line diagram of the two-bus power system~\eqref{eq:nl_power_dae}. Dynamic states $\x_d$ are the generator rotor angle and speed; algebraic states $\x_a$ are the load bus voltage angle and magnitude, coupled through the power flow constraints~\eqref{eq:nl_power_dae_phi1}-\eqref{eq:nl_power_dae_phi2}.}
\label{fig:two_bus_circuit}
\end{figure}

The safety specification enforces voltage boundedness at the load bus, requiring $V_\ell\in[V_{\min},V_{\max}]$ to prevent voltage violations that can trigger protective relay actions. Defining $h(\x):=(V_\ell-V_{\min})(V_{\max}-V_\ell)$, the safe set is $\mc{C}=\{\x:h(\x)\geq 0\}$. We consider a quadratic candidate barrier function
\begin{align}\label{eq:nl_power_barrier}
b(\x) :=& \beta_0-\beta_1(\delta_g-\delta_g^\star)^2-\beta_2\omega_g^2 \notag\\
&-\beta_3(\delta_\ell-\delta_\ell^\star)^2-\beta_4(V_\ell-V_\ell^\star)^2,
\end{align}
with coefficients $\beta_i>0$ selected around an equilibrium $(\delta_g^\star,\delta_\ell^\star,V_\ell^\star)$, so $\mc{D}=\{\x:b(\x)\geq 0\}$.

\begin{figure*}[t]
\centering
\includegraphics[width=0.85\textwidth]{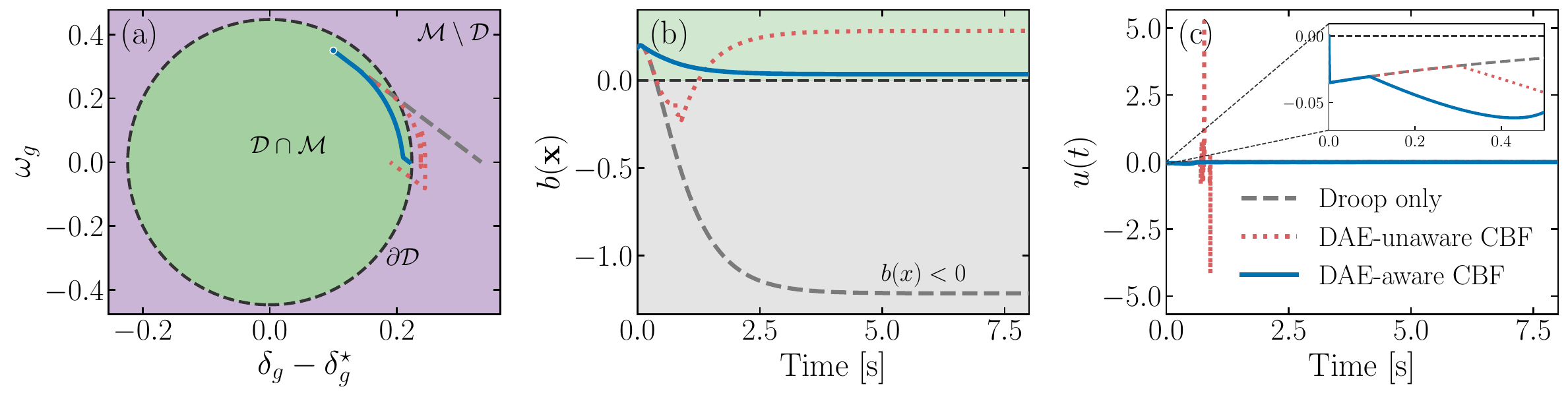}
\vspace{-0.25cm}
\caption{A power system DAE~\eqref{eq:nl_power_dae} with droop nominal controller $\ci_{\mr{nom}}=-K_p\omega_g$, $K_p=0.1$. \textbf{(a)}~Phase portrait in dynamic states $(\delta_g-\delta_g^\star,\omega_g)$: the DAE-aware CBF (blue) stays inside the barrier boundary $\partial\mc{D}\cap\mc{M}$, while both the DAE-unaware CBF (red, dotted) and droop-only (gray, dashed) trajectories exit $\mc{D}$. \textbf{(b)}~Barrier function $b(\x)$: the DAE-aware filter maintains $b\geq 0$; the droop-only system reaches $b=-1.22$; the DAE-unaware CBF drives $b$ to $-0.23$. \textbf{(c)}~Control inputs: the DAE-aware CBF applies modest corrections ($\max|\ci|=0.061$), while the DAE-unaware CBF produces destabilizing inputs ($\max|\ci|=5.26$).}
\label{fig:sim_nl_benchmark}
\vspace{-0.6cm}
\end{figure*}
Note that the derivation of candidate DAE-aware CBFs is outside the scope of this paper. The two-bus model isolates the generator-load algebraic coupling that is the building block of larger network-preserving DAE models, which makes it a canonical benchmark for validating DAE safety methods before scaling to higher-dimensional systems.

The parameters are $M=0.1$, $D=0.05$, $E_g=1.05$, $P_m=0.8$, with a symmetric two bus admittance matrix $Y_{gg}=Y_{\ell\ell}=5-j15$ and $Y_{g\ell}=Y_{\ell g}=-5+j15$. The load setpoints are $P_\ell^{\mr{spec}}=-0.78$, $Q_\ell^{\mr{spec}}=-0.56$, yielding the equilibrium $\delta_g^\star\approx 0.034$, $\omega_g^\star=0$, $\delta_\ell^\star=0$, $V_\ell^\star=1.0$. The voltage limits are $V_{\min}=0.9$ and $V_{\max}=1.1$, and the barrier coefficients are $(\beta_0,\ldots,\beta_4)=(1,10,5,10,100)$. All computations use Drake~1.40.0~\cite{drake} with MOSEK~10.2 on an Apple~M1~Pro with 16\,GB RAM.

Note that the barrier~\eqref{eq:nl_power_barrier} has relative degree $d=1$ with respect to the projected dynamics~\eqref{eq:projected_fields}. To see this, note that $g_d(\x)=[0,\;1/M]^\top$ enters $\dot{\omega}_g$ directly via~\eqref{eq:nl_power_dae_omega} and $b(\x)$ depends on $\omega_g$ through the term $-\beta_2\omega_g^2$; since $\phi$ is independent of $\omega_g$, we have $J_d(\x)g_d(\x)=0$ and the lower block of $\hat{g}$ in~\eqref{eq:projected_fields} vanishes, yielding $\mc{L}_{\hat{g}}b(\x)=\nabla b(\x)^\top\hat{g}(\x)=-2\beta_2\omega_g/M\neq 0$ near $\partial\mc{D}\cap\mc{M}$. Hence $d=\nu+d'-1=1$ (Definition~\ref{def:relative_degree}), and Theorem~\ref{thm:positive_invariance} applies directly without the HOCBF hierarchy of Corollary~\ref{cor:hocbf}. On the operating set, the Jacobian $J_a(\x)=\nabla_{\x_a}\phi(\x)\in\Rn{2\times 2}$ is nonsingular, satisfying Assumption~\ref{assumption:regularity}. Although this is a mild assumption, we verify regularity by sampling $1000$ points in $\mc{D}\cap\mc{M}$ and find $\min|\det J_a|=249.1$. The nominal controller is a proportional frequency droop $\ci_{\mr{nom}}=-K_p\omega_g$ with $K_p=0.1$. We simulate the DAE-aware CBF-QP with class-$\mc{K}_\infty$ function $\alpha(r)=r$ and initial condition $(\delta_g(0),\omega_g(0))=(\delta_g^\star+0.1\;\mr{rad},\,0.35\;\mr{rad/s})$ (near the boundary of $\mc{D}$, $b(\x(0))=0.188$), using backward Euler integration with step size $\Delta t=0.002$\,s over a horizon of $T=8$\,s ($4000$ steps). With the DAE-aware CBF filter active, $\min_t b(\x(t))\geq 3.5\times 10^{-2}>0$ and $\max_t|\ci(t)|=0.061$, confirming that $\mc{D}\cap\mc{M}$ is positively invariant. For comparison, we also run a \emph{DAE-unaware} CBF-QP that evaluates the Lie derivatives along the original $(f_d,g_d)$ instead of the projected $(\hat{f},\hat{g})$, ignoring the algebraic response $\dot{\x}_a=-J_a^{-1}J_d(f_d+g_d\ci)$. The DAE-unaware filter produces $\min_t b(\x(t))=-0.23$ and $\max_t|\ci(t)|=5.26$, violating safety despite applying large control effort. The droop-only case also violates safety ($\min_t b=-1.22$), as shown in Fig.~\ref{fig:sim_nl_benchmark}. This confirms that ignoring the algebraic coupling leads to safety violations, thereby answering $(\mr{Q}1)$.

We now verify the correctness and feasibility conditions in Theorem~\ref{thm:verification} according to the proposed SOS framework. We adopt the standard polynomial lifting used for power system DAEs~\cite{Kundu2015} as follows.
\begin{equation*}
z_1:=\sin(\delta_g-\delta_\ell),\qquad z_2:=\cos(\delta_g-\delta_\ell).
\end{equation*}
This replaces the trigonometric terms in the power flow equations with polynomial variables. To ensure that $(z_1,z_2)$ remains consistent with a physical angle difference, we add the algebraic constraint
\vspace{-0.2cm}\begin{equation}\label{eq:nl_power_lift_constraint} 
0=\phi_3(\x,z):=z_1^2+z_2^2-1,
\end{equation}
where the lifted state and constraint vector are defined as $\bar{\x}:=[\x^\top,z_1,z_2]^\top\in\Rn{n_x+2}$ and $\bar{\phi}(\bar{\x}):=[\phi_1(\x),\phi_2(\x),\phi_3(\x,z)]^\top$. Then $\phi_1,\phi_2$ become polynomial in $(\x,z)$. Thus, substituting $z_1,z_2$ into~\eqref{eq:nl_power_dae} yields a polynomial DAE on the lifted manifold $\bar{\mc{M}}:=\{\bar{\x}:\bar{\phi}(\bar{\x})=0\}$.

For correctness (Lemma~\ref{lem:correctness}), the choice $\beta_4=1/(h(\x^\star))=100$ admits the analytical Psatz identity~\eqref{eq:correctness_identity} with $\sigma_1=h(\x^\star)=0.01$, $\sigma_2=1$, $\sigma_0=\sigma_1[\beta_1(\delta_g-\delta_g^\star)^2+\beta_2\omega_g^2+\beta_3(\delta_\ell-\delta_\ell^\star)^2]$, $\rho_\ell=0$, and $k=1$. This certificate uses no algebraic constraint multipliers: $b(\x)\geq 0$ directly implies $(V_\ell-V_\ell^\star)^2\leq h(\x^\star)$ and hence $h(\x)\geq 0$. Correctness is also confirmed numerically by solving~\eqref{eq:correctness_identity} as an SDP with degree-$0$ multipliers $\sigma_1$ and $\rho_\ell$ on the six-variable lifted model ($n_x+2=6$ variables, $n_m=3$ constraints); MOSEK returns a feasible certificate in $0.02$\,s. For feasibility (Proposition~\ref{prop:sos_certificates}), since $\ci$ is unconstrained and $P(\x)=0$, the DAE-aware CBF condition reduces to $\mc{L}_{\hat{f}}b(\x)+\mc{L}_{\hat{g}}b(\x)\ci\geq -\alpha(b(\x))$, which admits a solution whenever $\mc{L}_{\hat{g}}b(\x)\neq 0$. Sampling $282$ points near $\partial\mc{D}\cap\mc{M}$, we find $\min|\mc{L}_{\hat{g}}b|=0.34>0$, confirming feasibility. This answers $(\mr{Q}2)$ for the nonlinear power DAE.

\subsection{Verification of Voltage Constraints}\label{subsec:pedersen_benchmark}
We now consider the two bus system in~\cite{Pedersen2016}, where bus~1 is a slack bus with fixed voltage magnitude and angle $(V_1,\delta_1)=(1,0)$, and bus~2 is a load bus.
Unless stated otherwise, the model equations, parameter values, voltage limits, and certificate in this subsection follow~\cite{Pedersen2016}.
The active and reactive load powers at bus~2 are modeled by the following linear dynamic model
\begin{subequations}\label{eq:pedersen_model}
\begin{align}
\dot{x} &=
\begin{bmatrix}
0 & \omega & 0\\
-\omega & 0 & 0\\
0 & 0 & -1/\tau
\end{bmatrix}x, \label{eq:pedersen_model_dyn}\\
\begin{bmatrix}p_2\\ q_2\end{bmatrix} &=
\begin{bmatrix}
c_{11} & 0 & c_{13}\\
c_{21} & 0 & c_{23}
\end{bmatrix}x+
\begin{bmatrix}d_1\\ d_2\end{bmatrix}. \label{eq:pedersen_model_out}
\end{align}
\end{subequations}
The parameter values are $\omega=2\pi/24$, $\tau=400$, $c_{11}=0.045$, $c_{13}=0.015$, $c_{21}=0.008$, $c_{23}=0.015$, $d_1=-0.83$, and $d_2=-0.62$.

For the polynomial DAE form, define algebraic variables
\vspace{-0.1cm}
\begin{equation*}
z_1:=V_2,\qquad z_2:=\sin(\delta_2),\qquad z_3:=\cos(\delta_2).
\end{equation*}
Let $x_d:=x\in\Rn{3}$, $x_a:=z\in\Rn{3}$, and $\x=[x_d^\top,x_a^\top]^\top\in\Rn{6}$. 

The algebraic constraints are
\begin{subequations}\label{eq:pedersen_algebraic}
\begin{align}
0 &= \phi_1(\x):= p_2-\psi_p(z), \\
0 &= \phi_2(\x):= q_2-\psi_q(z), \\
0 &= \phi_3(\x):= z_2^2+z_3^2-1,
\end{align}
\end{subequations}
where $\psi_p,\psi_q$ denote the active and reactive power flow expressions written in lifted polynomial coordinates, induced by the two bus admittance matrix in~\cite{Pedersen2016}; we use voltage limits $V_{\min}=0.9$ and $V_{\max}=1.1$.

The safety specification is voltage boundedness at bus~2, defined by $h(\x):=(z_1-V_{\min})(V_{\max}-z_1)$, with safe set $\mc{C}=\{\x:h(\x)\geq 0\}$. We consider the following candidate barrier function from~\cite{Pedersen2016}
\begin{equation}\label{eq:pedersen_barrier}
b(\x):=8.52-0.416x_1^2-0.416x_2^2-0.171x_3^2,
\end{equation}
defining the barrier set $\mc{D}=\{\x:b(\x)\geq 0\}$.

We apply Corollary~\ref{cor:autonomous_index1} to~\eqref{eq:pedersen_model}-\eqref{eq:pedersen_algebraic}. Since $n_m=n_a=3$, the Jacobian $J_a(\x)\in\Rn{3\times 3}$ is square and $P(\x)=0$ (Remark~\ref{remark:reduction_to_ode}), so the compatibility condition~\eqref{eq:autonomous_compat} is trivially satisfied. Sampling $2000$ points in $\mc{D}\cap\mc{M}$, we find $\min|\det J_a|=500.9$, confirming regularity (Assumption~\ref{assumption:regularity}). For the CBF condition~\eqref{eq:autonomous_cbf}, note that $b$ depends only on $x_d$ (i.e., $\nabla_{x_a}b=0$), thus the projected Lie derivative reduces to $\mc{L}_{\hat{f}}b(\x)=\nabla_{x_d}b\cdot Ax_d$. The skew-symmetric oscillator block of $A$ cancels the $x_1$-$x_2$ cross terms, yielding $\mc{L}_{\hat{f}}b(\x)=(2\beta_3/\tau)x_3^2\geq 0$ for all $\x$, where $\beta_3=0.171$ is the coefficient of $x_3^2$ in~\eqref{eq:pedersen_barrier} and $\tau=400$. Since $\mc{L}_{\hat{f}}b\geq 0\geq -\alpha(b)$ on $\mc{D}$ for any $\alpha\in\mc{K}_\infty$, condition~\eqref{eq:autonomous_cbf} holds and Corollary~\ref{cor:autonomous_index1} certifies positive invariance of $\mc{D}\cap\mc{M}$. Correctness ($\mc{D}\cap\mc{M}\subseteq\mc{C}\cap\mc{M}$) is confirmed by sampling: $\min_{\mc{D}\cap\mc{M}}h(\x)=4.2\times 10^{-3}>0$. This answers $(\mr{Q}3)$ and concludes the case studies.

\section{Conclusion}\label{sec:conclusion}
This paper develops a DAE-aware CBF framework for regular semi-explicit DAE systems. We introduce a projected dynamics representation on the constraint manifold, which yields safety conditions that enforce set invariance together with algebraic consistency. We then provide SOS verification conditions for correctness and feasibility of candidate barrier functions. Power system case studies illustrate the proposed results and show that accounting for algebraic coupling is essential for reliable safety guarantees.

\balance
\bibliographystyle{IEEEtran}
\bibliography{ref}
\end{document}

%% file: usepackages_letter.tex
\usepackage[hidelinks]{hyperref}
\hypersetup{
        pdftitle={Safe control of semi-explicit descriptor systems},
        pdfauthor={Mohamad H. Kazma},
}

\usepackage{resizegather}

\usepackage{cite}
\usepackage{amsmath,amssymb,amsfonts,mathrsfs}
\usepackage{amsthm}
\usepackage{graphicx}
\usepackage{mdframed}
\usepackage{textcomp}
\usepackage{textcomp}
\usepackage{tablefootnote}
\usepackage{threeparttable}
\usepackage{bbm}
\usepackage{dsfont}
\usepackage{tikz}
\usepackage{graphicx}
\usepackage{tkz-euclide}
\usepackage{algorithm}
\usepackage{algpseudocode}
\usepackage{booktabs}
\usetikzlibrary{positioning}
\usetikzlibrary{shapes}
\usetikzlibrary{shapes.misc}
\usetikzlibrary{shapes.geometric}
\usetikzlibrary{plotmarks}
\usetikzlibrary{intersections}
\usetikzlibrary{calc}
\usetikzlibrary{fit}
\usetikzlibrary{patterns,tikzmark}
\usetikzlibrary{matrix,decorations.pathreplacing,calc}
\usetikzlibrary{backgrounds}

\tikzset{cross/.style={cross out, draw, 
         minimum size=2*(#1-\pgflinewidth), 
         inner sep=0pt, outer sep=0pt}}

 \usepackage{circuitikz} 
 
\usepackage{enumitem}
\setlist[itemize]{leftmargin=*}
\setlist[enumerate]{leftmargin=*}

%% file: macros.tex
\newcommand{\ci}[0]{u}

\newcommand{\x}[0]{x}

\newcommand{\normnsz}[1]{\|#1\|}
\newcommand{\vertiii}[1]{{\left\vert\kern-0.25ex\left\vert\kern-0.25ex\left\vert #1 
    \right\vert\kern-0.25ex\right\vert\kern-0.25ex\right\vert}}

\newtheorem{theorem}{Theorem}
\newtheorem{proposition}{Proposition}

\newtheorem{corollary}{Corollary}
\newtheorem{lemma}{Lemma}

\newtheorem{problem}{Problem}
\newtheorem{assumption}{Assumption}
\newtheorem{remark}{Remark}

\newtheorem{definition}{Definition}

\newmdenv[
    linecolor=white, backgroundcolor=lightgray!15, innertopmargin=5pt, innerbottommargin=5pt, skipabove=10pt, skipbelow=10pt
]{graybox}

\makeatletter
\renewcommand{\fps@figure}{htp}
\renewcommand{\fps@table}{htp}
\makeatother

\newcommand{\mc}[1]{\mathcal{#1}}
\newcommand{\mbb}[1]{\mathbb{#1}}
\newcommand{\mr}[1]{\mathrm{#1}}

\newcommand{\R}{\mathbb{R}}
\newcommand{\Rn}[1]{\mathbb{R}^{#1}}

\usepackage{comment}
\definecolor{blue(munsell)}{rgb}{0.0, 0.5, 0.69}

\newcommand{\titlesc}[1]{\title{\Large \centering {\textsc{{#1}}}}}

\newcommand{\parstart}[1]{\noindent \textbf{#1.}\;}